\documentclass[12pt]{article}
\usepackage{amsmath}
\usepackage{amsthm}
\usepackage{thmtools, thm-restate}
\usepackage{graphicx,psfrag,epsf}
\usepackage{enumerate}
\usepackage{natbib}
\usepackage{siunitx}
\usepackage{url} 
\usepackage{bbm}
\usepackage{rotating}
\usepackage{xcolor}
\usepackage{caption}
\usepackage{subcaption}
\usepackage{afterpage}

\usepackage{verbatim}

\usepackage[export]{adjustbox}
\usepackage{array}

\usepackage[linesnumbered,ruled]{algorithm2e}

\RequirePackage{amsmath, amssymb}

\declaretheorem{theorem}
\declaretheorem[sibling=theorem]{lemma}

\declaretheorem[sibling=theorem]{proposition}

\usepackage{amsmath}
   
\newif\ifJASA
\JASAfalse 

\newif\ifnotblind
\notblindtrue 

\ifJASA
\else

\fi

\begin{document}

\def\spacingset#1{\renewcommand{\baselinestretch}%
{#1}\small\normalsize} \spacingset{1}


{
  \title{\bf Poisson-FOCuS: An efficient online method for detecting count bursts with application to gamma ray burst detection.}
\ifnotblind
  \author{
    Kes Ward\\
    STOR-i Doctoral Training Centre,\\ Lancaster University, Lancaster, UK\\~\\
    Giuseppe Dilillo \\
    Department of Mathematics, Computer Science and Physics, \\University of Udine, Udine, Italy\\~\\
    Idris Eckley \\
    Department of Mathematics and Statistics,\\ Lancaster University, Lancaster, UK
     \\~\\
    Paul Fearnhead\\
    Department of Mathematics and Statistics,\\ Lancaster University, Lancaster, UK}
\fi
  \maketitle
}


\begin{abstract}
Gamma-ray bursts are flashes of light from distant exploding stars. Cube satellites that monitor photons across different energy bands are used to detect these bursts. There is a need for computationally efficient algorithms, able to run using the limited computational resource onboard a cube satellite, that can detect when gamma-ray bursts occur. Current algorithms are based on monitoring photon counts across a grid of different sizes of time window. We propose a new algorithm, which extends the recently developed FOCuS algorithm for online change detection to Poisson data. Our algorithm is mathematically equivalent to searching over all possible window sizes, but at half the computational cost of the current grid-based methods. We demonstrate the additional power of our approach using simulations and data drawn from the Fermi gamma-ray burst catalogue.

\end{abstract}

\noindent%
\vfill

\ifJASA
\spacingset{1.45} 
\fi

\section{Introduction}

This work is motivated by the challenge of designing an efficient algorithm for detecting GRBs for cube satellites, such as the HERMES scientific pathfinder \cite[]{Hermes-2021}, during the early 2020s. Cube satellites are compact and therefore relatively cheap to launch into space, but have limited computational power on-board. 

Gamma ray bursts (GRBs) are short-lived bursts of gamma-ray light. They were first detected by satellites in the late 1960s \cite[]{Klebesadel1973-sg}, and are now believed to arise when matter collapses into black holes. At the time of writing there is considerable interest in detecting gamma ray bursts with a viewe to identifying interesting astronomical features \cite[]{Luongo2021-ik}.

The HERMES satellite array is interested in monitoring a variety of different photon energy bands, and regions in the sky, for the appearance of a GRB. Raw data from a satellite consists of a data stream of photons impacting a detector. The time of a photon impact is recorded in units of microseconds since satellite launch. New photon impacts are recorded on the order of approximately one every 500 microseconds. A GRB is indicated by a short period of time with an unusually high incidence of photons impacting the detector. Ideally the satellite would detect each GRB, and for each burst it detects it then transmits the associated data to earth.

There are a number of statistical challenges with detecting GRBs. First, they can come from close or far away sources, and can therefore be very bright and obvious to observe or very dim and hard to pick out from background. They can also impact the detector over a variety of timescales. Figure \ref{fig:short_and_long_bursts} shows two GRBs, one short and intense lasting a fraction of a second, one longer and less intense lasting about ten seconds. These bursts were taken from the FERMI catalogue \cite[]{Ajello2019-gn}. Bursts ranging from a fraction of a second to a few minutes are possible.
Second, less than one burst is recorded per 24 hours on average \cite[]{von2020fourth}, which is relatively rare in comparison to the velocity of the signal.
The background rate at which the satellite detects photons also varies over time. This variation is on timescales much larger (minutes to days) than those on which bursts occur (milliseconds to seconds), and thus is able to be estimated separately from the bursts. We will not address this aspect of the problem directly, though Appendix \ref{section:background_bias} gives some guidance on choice of background estimation approaches.
Finally, there are also computational challenges. For example, there is limited computational hardware on board the satellite, and additional constraints arise on the use of these due to battery life and lack of heat dissipation \cite[]{2003fenimore}. There is also a substantial computational and energy cost to transmitting data to earth, so only promising data should be sent.
\begin{figure}
     \centering
     \begin{subfigure}[b]{0.49\textwidth}
         \centering
         \includegraphics[width=\textwidth]{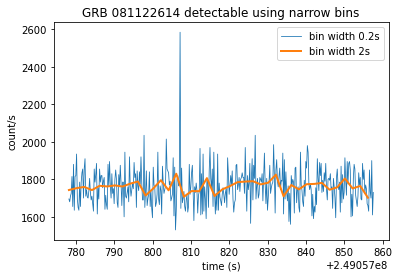}
         \caption{a short, intense burst}
         \label{subfig:short_burst}
     \end{subfigure}
     \hfill
     \begin{subfigure}[b]{0.49\textwidth}
         \centering
         \includegraphics[width=\textwidth]{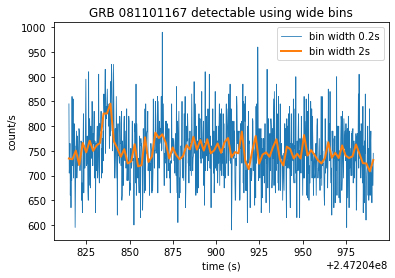}
         \caption{a long, less intense burst}
         \label{subfig:long_burst}
     \end{subfigure}
     
     \caption{Plots of two recorded gamma ray bursts from the FERMI catalogue, with photon counts binned into $0.2$s and $2$s intervals.}
     \label{fig:short_and_long_bursts}
\end{figure}

At a fundamental level, algorithmic techniques for detecting GRBs have gone unchanged through different generations of space-born GRB monitor experiments. As they reach a detector, high-energy photons are counted over a fundamental time interval and in different energy bands. Count rates are then compared against a background estimate over a number of pre-defined timescales \cite[]{lima1983}. To minimize the chance of missing a burst due to a mismatch between the event activity and the length of the tested timescales, multiple different timescales are simultaneously evaluated. Whenever the significance in count excess is found to exceed a threshold value over a timescale, a trigger is issued. 

Figure \ref{fig:computing_flowchart} gives a simplified overview of such a detection system. As data arrives we need to both detect whether a gamma ray burst is happening, and update our estimates of the background photon arrival rate. Because of the high computational cost of transmitting data to earth after a detection, if an algorithm detects a potential gamma ray burst there is an additional quality assurance step to determine whether it should be transmitted. The detection algorithm needs to be run at a resolution at which all gamma ray bursts are detectable. By comparison background re-estimation is only required once every second, and the quality assurance algorithm is only needed every time a potential gamma ray burst is detected. Thus the majority of the computational effort is for the detection algorithm -- and how to construct a statistically efficient detection algorithm with low computation is the focus of this paper.

\begin{figure}[t!]
    \centering
    \includegraphics[width=0.9\textwidth]{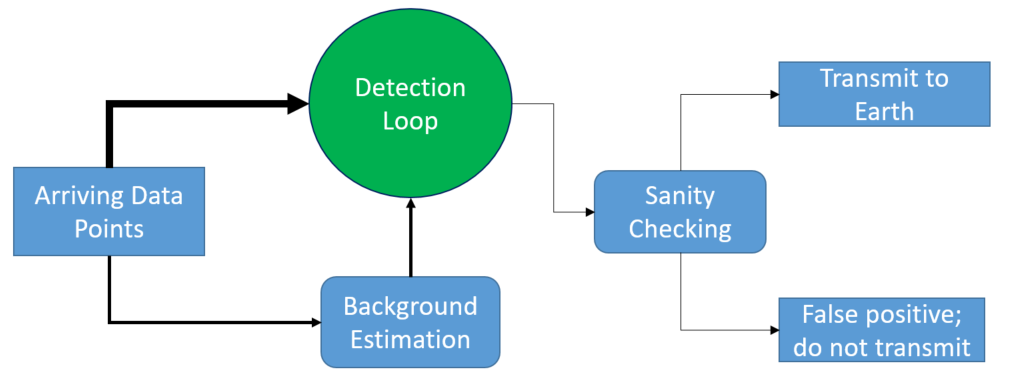}
    \caption{A schematic of the detection system, with the arrow thickness corresponding to the relative velocities of data flows. Most of the computing requirements of the trigger algorithm are within the detection loop, highlighted in green, which is the focus of this paper.}
    \label{fig:computing_flowchart}
\end{figure}

As mentioned, current practice for detecting a GRB is to compare observed photon counts with expected counts across a given bin width \cite[]{2012paciesas}. The choice of bin width affects the ease of discovery of different sizes of burst. Figure \ref{fig:short_and_long_bursts} shows an example. The short burst in Figure \ref{subfig:short_burst} is easily detectable with bin width $0.2$s, but lost to smoothing at a $2$s bin width. In contrast, the burst in Figure \ref{subfig:long_burst} has a signal too small relative to the noise to be detectable at a $0.2$s bin width, with the largest observation on the plot being part of the noise rather than the gamma ray burst. Only when smoothed to a bin width of $2s$ does the burst become visually apparent. Therefore, the bin width is first chosen small enough to pick up short bursts, and geometrically spaced windows of size 1, 2, 4, 8, ... times the bin width, up to a maximum window size, are run over the data in order to pick up longer bursts. 

This paper develops an improved approach to detecting GRBs. First we show that using the Page-CUSUM statistic \cite[]{Page1954-ej,Page1955-zf}, and its extension to count data \cite[]{Lucas1985-gm} is uniformly better than using a window-based procedure. These schemes require specifying both the pre-change and post-change behaviour of the data - in our case specifying the background rate of photon arrivals and the rate during the gamma ray burst. While in our application it is reasonable to assume that good estimates of the background photon arrival rate are available by monitoring the signal, specifying the photon arrival rate for the gamma ray burst is difficult due to their heterogeneity in terms of intensity. For detecting changes in mean in Gaussian data, \cite{Romano2021-ab} show how one can implement the sequential scheme of \cite{Page1955-zf} simultaneously for all possible post-change means, and call the resulting algorithm FOCuS. Our detection algorithm involves extending this algorithm to the setting of detecting changes in the rate of events for count data. It is based on modelling the arrival of photons on the detector as a Poisson process, and we thus call our detection algorithm Poisson-FOCuS.

Our algorithm is equivalent to checking windows of any length, with a modified version equivalent to checking windows of any length up to a maximum size. This makes it advantageous for detecting bursts near the chosen statistical threshold whose length is not well described by a geometrically spaced window. In addition to this, the algorithm we develop has a computational cost lower than the geometric spacing approach, resulting in a uniform improvement on the methods already used for this application with no required trade-off. These advantages mean that the Poisson-FOCuS algorithm is currently planned to be used as part of the trigger algorithm of the HERMES satellite.

Our improvement of existing window based methods addresses the aspect of trigger algorithms that has been shown to be most important for increasing power of detecting GRBs. As the computational resource on-board a satellite has increased, trigger algorithms have grown to support an increasing number of criteria, and is has been seen that the most important aspect of any detection procedure is the timescale over which the data is analyzed \cite[]{2004mclean}. For example, while early software, such as Compton-BATSE \cite[]{Gehrels1993-ep}, operated only a few different trigger criteria, a total of $120$ are available to the Fermi-GBM \cite[]{2009meegan} and and more than $800$ to the Swift-BAT \cite[]{2004mclean} flight software. While in many cases, this growth in algorithm complexity did not result in more GRB detection, better coverage of different timescales for GRBs did \cite[]{2012paciesas}. During the first four years of Fermi-GBM operations, $135$ out of $953$ GRBs triggered GBM over timescales not represented by BATSE algorithms, most of which were over timescales larger than the maximum value tested by BATSE ($1.024$ s) \cite[]{2014kienlin}. 

The outline of the paper is as follows. In Section \ref{section:problem_setup} we define the mathematical setup of the problem and analyse the statistical models used in current and possible alternate approaches. Section \ref{section:functional_pruning} provides the main theoretical developments of a functional pruning approach, leading to an algorithm and computational implementation specified in Section \ref{section:theoretical_evaluation}. In Section \ref{section:empirical_evaluation} we give an evaluation of our method on various simulated data, and real data taken from the FERMI catalogue. The paper ends with a discussion. All proofs, and the calculations necessary for graphs, can be found in the Appendix.

\section{Modelling framework}
\label{section:problem_setup}

The data we consider takes the form of a time-series of arrival times of photons.  We can model the generating process for these points as a Poisson process with background parameter $\lambda(t)$, and with gamma ray bursts corresponding to periods of time which see a increase in the arrival rate over background level. 
Changes in the background rate $\lambda(t)$ over time may be due to rotation of the spacecraft or its orbit around the earth. They exist on a greater timescale (minutes to days) than the region of time over which an anomalous interval could occur (seconds). 

We assume that a good estimate of the current background rate $\lambda(t)$ is available. To ease exposition we will first assume this rate is known and constant, and denote it as $\lambda$ before generalising to the non-constant background rate in Section \ref{section:varying_background}. For readers interested in applying our method, we also discuss accounting for error in estimating the background rate in Appendix \ref{section:background_bias}.

There are two approaches to analysing our data. The first is to choose a small time interval, $w$, which is smaller than the shortest gamma ray burst that we want to detect. In this case the data can be summarised by the number of photon arrivals in time bins of length $w$. By rescaling our time unit measurement (and therefore our rate $\lambda$) appropriately we can set $w=1$ without loss of generality. The second approach is to use the data of times between photon arrivals directly. We will explain how our detection algorithm can be implemented in this latter setting in Section \ref{section:other_data}.

For the first setting, we will denote the the data by $x_1,x_2,\ldots,$ with $x_i$ denoting the number of arrivals in the $i$th time window. We use the notation $x_{t+1:t+h}=(x_{t+1},\ldots,x_{t+h})$ to denote the vector of observations between the $(t+1)$th and $(t+h)$th time window, and 
\[
\bar{x}_{t+1:t+h}=\frac{1}{h}\sum_{i=t+1}^{t+h} x_{i},
\]
the mean of these observations.
Under our model, if there is no gamma ray burst then each $x_i$ is a realisation of $X_i$, an independent Poisson distribution with parameter $\lambda$. If there is a gamma ray burst then the number of photon arrivals will be Poisson distributed with a rate larger than $\lambda$. We make the assumption that a gamma ray burst can be characterised by a width, $h$, and an intensity $\mu>1$ such that if the gamma ray burst starts at time $t+1$ then $x_{t+1},\ldots,x_{t+h}$ are realisations of independent Poisson random variables $X_{t+1},\ldots,X_{t+h}$ with mean $\mu\lambda$. See Figure \ref{fig:two_dimensions_anomalies_poisson} for a visualisation of an anomaly simulated directly from this model. 

Our algorithm is primarily interested in reducing the computation requirements of constant signal monitoring. Therefore our model considers a gamma ray burst as a uniform increase in intensity over its length, which does not take into account the unknown shape of a gamma ray burst. If a possible burst is found, an additional round of shape-based sanity checking requiring more computational resources can easily be performed prior to transmission back to Earth.

\begin{figure}[t]
    \centering
    \includegraphics[width=0.5\textwidth]{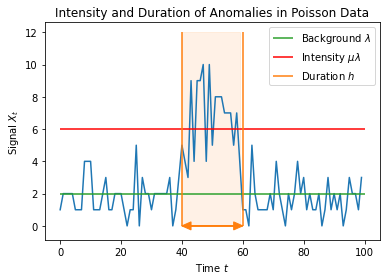}
    \caption{A simulated example anomaly with intensity multiplier $\mu=3$ and duration $h=20$ against a background $\lambda=2$. 
    }
    \label{fig:two_dimensions_anomalies_poisson}
\end{figure}

\subsection{Window-based methods and detectability}


If we assumed we knew the width of the gamma ray burst, $h$, then detecting it would correspond to testing, for each start time $t$, between the following two hypotheses:
\begin{itemize}
    \item $\mathbf{H}_0$: $X_{t+1},\ldots,X_{t+h} \sim \text{Poisson}(\lambda)$.
    \item $\mathbf{H}_1$: $ X_{t+1},\ldots,X_{t+h} \sim \text{Poisson}(\mu \lambda)$, for some $\mu>1$.
\end{itemize}

\noindent We can perform a likelihood ratio test for this hypothesis. Let $\ell(x_{t+1:t+h};\mu)$ denote the log-likelihood for the data $x_{t+1:t+h}$ under our Poisson model with rate $\mu\lambda$. Then the standard likelihood ratio statistic is
\[
LR=2\left\{\max_{\mu>1} \ell(x_{t+1:t+h};\mu)-\ell(x_{t+1:t+h};1) \right\}.
\]

\begin{restatable}{proposition}{lrderivation}

Under the alternative, our LR statistic is 0 if $\bar{x}_{t+1:t+h} \leq \lambda$, otherwise 

\[LR =  2h\lambda \left\{\frac{\bar{x}_{t+1:t+h}}{\lambda} \log \left(\frac{\bar{x}_{t+1:t+h}}{\lambda}\right) -  \left(\frac{\bar{x}_{t+1:t+h}}{\lambda}-1\right) \right\}. \]

\end{restatable}

The likelihood ratio is a function only of the expected count $h\lambda$ and the fitted intensity $\hat{\mu}_{t+1:t+h} := \bar{x}_{t+1:t+h}/\lambda$ of the interval $[t+1,t+h]$. 
It can alternatively be written as a function only of the expected count $h\lambda$ and the actual count $h\bar{x}_{t+1:t+h}$, which forms the fundamental basis for our algorithm.

\begin{figure}
     \centering
     \begin{subfigure}[b]{0.5\textwidth}
         \centering
         \includegraphics[width=\textwidth]{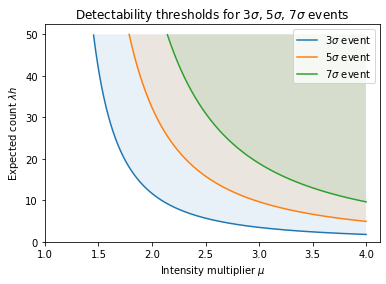}
     \end{subfigure}
     
     \caption{Detectability of GRBs at different $k$-sigma levels. Shaded regions show values of $h\lambda$ and $\hat{\mu}_{t+1:t+h}$ where the likelihood ratio exceeds $k$-sigma event thresholds for $k=3$ (blue region), $k=5$ (orange region) and $k=7$ (green region) for a test that uses the correct value of $h$.}
     \label{fig:detectability_thresholds}
\end{figure}

In our application, thresholds for gamma ray burst detection are often set based on $k$-sigma events: values that are as extreme as observing a Gaussian observation that is $k$ standard deviations from its mean. As the likelihood ratio statistic is approximately $\chi^2_1$ distributed \cite[]{Wilks1938-lh}, this corresponds to a threshold of $k^2$.

Gamma ray bursts with a combination of high intensity $\hat{\mu}_{t+1:t+h} := \frac{\bar{x}_{t+1:t+h}}{\lambda}$ and long length, as quantified by the expected count $h\lambda$, will have higher associated likelihood ratio statistics and thus be easier to detect. Figure \ref{fig:detectability_thresholds} shows regions in this two-dimensional space that correspond to detectable GRBs at different $k$-sigma levels. 
Figure \ref{fig:window_detectability} shows what happens when we set a fixed threshold and for computational reasons only check intervals of certain lengths. We rely on the fact that a slightly brighter burst will also trigger detection on a longer or shorter interval than optimal. This is the type of approach that current window-based methods take \cite[]{2012paciesas}. We can see that, even with a grid of window sizes, we lose detectability if the true width of the GRB does not match one of the window sizes.

\begin{figure}
     \centering
     \begin{subfigure}[b]{0.49\textwidth}
         \centering
         \includegraphics[width=\textwidth]{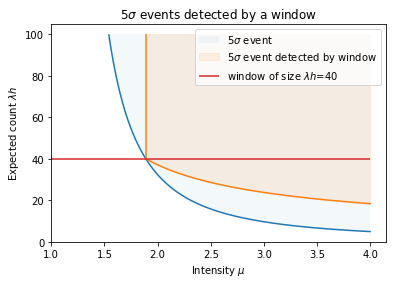}
     \end{subfigure}
     \hfill
     \begin{subfigure}[b]{0.49\textwidth}
         \centering
         \includegraphics[width=\textwidth]{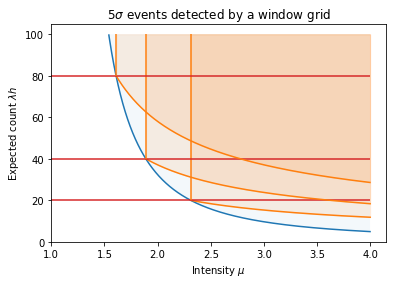}
     \end{subfigure}
     
     \caption{Detectability of GRBs by the window method, for one window (left) or a grid of three windows (right). Orange shaded area shows the values of $h\lambda$ and $\hat{\mu}_{t+1:t+h}$ where the likelihood ratio exceeds a 5-sigma threshold for the window test, and blue shaded area shows the detectability region form Figure \ref{fig:detectability_thresholds}. }
     \label{fig:window_detectability}
\end{figure}

\subsection{Page-CUSUM for Poisson data}

As a foundation for our detection algorithm, consider the CUSUM (cumulative sum) approach of \cite{Page1955-zf} that was adapted for the Poisson setting by \cite{Lucas1985-gm}. These methods search for gamma ray bursts of unknown width but known size $\mu$, differing from a window method that searches for gamma ray bursts of known width $h$ and unknown size. To run our methods online it is useful to characterise by the start point $\tau$ of a possible anomaly.
We have our hypotheses for the signal at time $T$:
\begin{itemize}
    \item $\mathbf{H}_0$: There have been no anomalies, i.e. $X_1, ..., X_T \sim \text{Poisson}(\lambda)$.
    \item $\mathbf{H}_1$: There has been one anomaly, beginning at some unknown time $\tau$, with known intensity multiplier $\mu>1$, i.e. $X_1, ..., X_{\tau-1} \sim \text{Poisson}(\lambda)$ and $X_{\tau}, ..., X_T \sim \text{Poisson}(\mu \lambda)$.
\end{itemize}

\begin{restatable}{proposition}{lrderivationpage}

Under the alternative, our LR statistic is 0 if $\bar{x}_{\tau:T} \leq \lambda \frac{\mu-1}{\log(\mu)}$ for all $\tau$, otherwise 

\[LR =  \max_{1 \leq \tau \leq T} \left[2(T-\tau+1)\lambda \left\{\frac{\bar{x}_{\tau:T}}{\lambda} \log \left(\mu \right) -  \left(\mu-1\right) \right\} \right]. \]

\end{restatable}

We work with a test statistic $S_T$ that is half the likelihood ratio statistic for this test, and compare to a $k$-sigma threshold of $k^2/2$. $S_T$ can be rewritten in the following form:
$$S_T = \left[\max_{1 \leq \tau \leq T} \sum_{t=\tau}^T(x_t \log (\mu) - \lambda (\mu-1))\right]^+,$$
where we use the notation $[\cdot]^+$ to denote the maximium of the term $\cdot$ and 0. As shown in \cite{Lucas1985-gm}, $S_T$ can be updated recursively as 
$$S_0 = 0, \ \ S_{T+1} = [S_T + x_{T+1} \log \mu - \lambda(\mu-1)]^+.$$ 

It is helpful to compare the detectability of GRBs using $S_T$ with their detectability using a window method. To this end, we introduce the following propositions that ....

\begin{restatable}{proposition}{pageonlywindow}
\label{prop:page_only_window}
For some choice of $\mu$ against a background rate of $\lambda$, let $S_T$ be significant at the $k$-sigma level. Then there exists some interval $[\tau,T]$ with associated likelihood ratio statistic that is significant at the $k$-sigma level.
\end{restatable}

\begin{restatable}{proposition}{pagebeatswindow}
\label{prop:page_beats_window}
For any $k$, $\lambda$ and $h$ there exists a $\mu$ and corresponding test statistic $S_T$ that relates directly to a window test of length $h$ and background rate $\lambda$ as follows: if, for any $t$, the data $x_{t+1:t+h}$ is significant at the $k$-sigma level then $S_{t+h}$ will also be significant at the $k$-sigma level.

\end{restatable}

Together these results show that Page's method is at least as powerful as the window method. Rather than implementing the window method with a given window size, we can implement Page's method with the appropriate $\mu$ value (as defined by Proposition \ref{prop:page_beats_window}) such that any GRB detected by the window method would be detected by Page's method.  
However Page's method may detect additional GRBs and these would be detected by the window method with some window size (by Proposition \ref{prop:page_only_window}). In practice, as shown in Figure \ref{fig:page_window_comparison}, Page's method provides better coverage of the search space.

\begin{figure}
     \centering
     \begin{subfigure}[b]{0.49\textwidth}
         \centering
         \includegraphics[width=\textwidth]{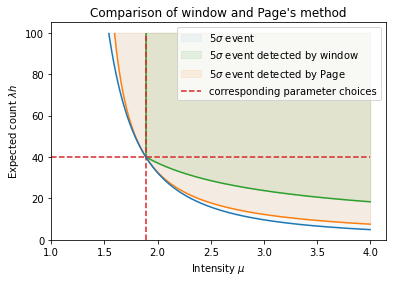}
         \caption{}
         \label{fig:page_window_comparison}
     \end{subfigure}
     \hfill
     \begin{subfigure}[b]{0.49\textwidth}
         \centering
         \includegraphics[width=\textwidth]{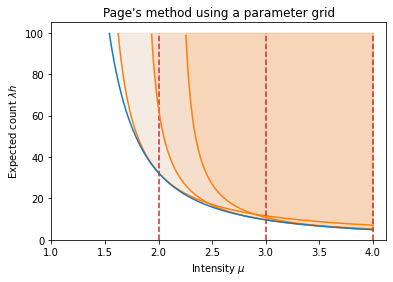}
         \caption{}
         \label{fig:page_grid}
     \end{subfigure}
     
     \caption{Detectability of GRBs by Page-CUSUM for a single $\mu$ value (left) and a grid of $\mu$ values (right).  Orange shaded area shows the values of $h\lambda$ and $\hat{\mu}_{t+1:t+h}$ where the likelihood ratio exceeds a 5-sigma threshold for the Page-CUSUM ($\mu$ value shown by red vertical line). For comparison, the green shaded area shows the detectability region for the corresponding window test as defined by Proposition \ref{prop:page_beats_window} (left-hand plot), and blue shaded area shows detetability region from Figure \ref{fig:detectability_thresholds}.}
     \label{fig:pages_method}
\end{figure}

Whilst the Page-CUSUM approach is more powerful than a window-based approach, to cover the space completely it still requires specifying a grid of values for the intensity of the gamma ray burst as in Figure \ref{fig:page_grid}. If the actual intensity lies far from our grid values we will lose power at detecting the burst. 

\section{Functional pruning}
\label{section:functional_pruning}

To look for an anomalous excess of count of any intensity and width without having to pick a parameter grid, we consider computing the Page-CUSUM statistic simultaneously for all $\mu\in [1, \infty)$, which we can do by considering the test statistic as a function of the $\mu$,  $S_T(\mu)$. That is, for each $T$, $S_T(\mu)$ is defined for $\mu\in [1, \infty)$, and for a given $\mu$ is equal to value of the Page-CUSUM statistic for that $\mu$. 

By definition, $S_T(\mu)$ is a pointwise maximum of curves representing all possible anomaly start points $\tau$:
$$S_T(\mu) := \left[\max_{1 \leq \tau \leq T} \sum_{t=\tau}^T\left[x_t \log (\mu) - \lambda (\mu-1)\right]\right]^+$$
We can view this as $S_T(\mu)=[\max_{1 \leq \tau \leq T} C^{(T)}_\tau(\mu)]^+$, where each curve $C^{(T)}_{\tau}(\mu)$ corresponds to half the likelihood ratio statistic for a gamma ray burst of intensity $\mu$ starting at $\tau$,
$$C^{(T)}_{\tau}(\mu) := \sum_{t=\tau}^T[x_t \log (\mu) - \lambda (\mu-1)] .$$

\noindent Each curve is parameterised by two quantities, as
$$C^{(T)}_{\tau}(\mu) := a^{(T)}_{\tau}\log (\mu) - {b^{(T)}_\tau}(\mu-1),$$
where $a^{(T)}_{\tau} = \sum_{t=\tau}^T x_t$ is the actual observed count and $b^{(T)}_{\tau} = \sum_{t=\tau}^T \lambda = \lambda(T-\tau+1)$ is the expected count on the interval $[\tau, T]$. As we move from time $T$ to time $T+1$ there is a simple recursion to update these coefficients
$$ a^{(T+1)}_{\tau} =a^{(T)}_{\tau} + x_{T+1}, \ \ \ b^{(T+1)}_{\tau} = b^{(T)}_{\tau} + \lambda .$$
These are linear and do not depend on $\tau$, so the differences between any two curves are preserved with time updates.

\begin{figure}
     \centering
     \begin{subfigure}[b]{0.49\textwidth}
         \centering
         \includegraphics[width=\textwidth]{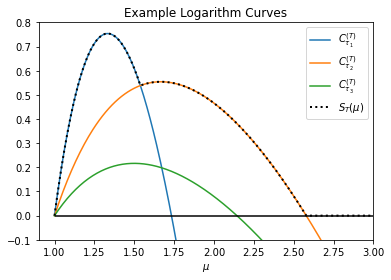}
     \end{subfigure}
     
     \caption{Three example logarithm curves. The statistic $S_T(\mu)$ is defined as the maximum of all logarithm curves and the $0$ line.}
     \label{fig:example_logarithm_curves}
\end{figure}

\subsection{Structure of logarithm curves}

We call $C^{(T)}_{\tau}(\mu)$ a logarithm curve.
Figure \ref{fig:example_logarithm_curves} shows the structure of what these logarithmic curves look like. Note that if the observed count $a_{\tau}^{(T)}$ is not greater than the expected count $b_{\tau}^{(T)}$, the curve $C_{\tau}^{(T)}$ will be non-positive on $\mu \in [1, \infty)$ as it contains no evidence for a change greater than the background rate $\lambda$ over $[\tau, T]$. The maximum of $C_{\tau}^{(T)}$ is located at $\mu=a_{\tau}^{(T)}/b_{\tau}^{(T)}$, representing the likelihood ratio for a post-change mean $\mu \lambda = \bar{x}_{\tau:T}$.

If $a_\tau^{(T)}>b_\tau^{(T)}$ then the logarithm curve will be positive for some $\mu>1$. In this case we will define the root of the curve to be the unique $\mu^*>1$ such that $C^{(T)}_{\tau}(\mu^*)=0$.

\subsection{Adding and pruning curves}

For any two curves $C^{(T)}_{\tau_i}$ and $C^{(T)}_{\tau_j}$ at a given present time $T$, we will say that $C^{(T)}_{\tau_i}$ dominates $C^{(T)}_{\tau_j}$ if 
$$[C^{(T)}_{\tau_i}(\mu)]^+ \geq [C^{(T)}_{\tau_j}(\mu)]^+,\ \ \ \forall \mu \in [1, \infty).$$
This is equivalent to saying that there is no value of $\mu$ such that the interval $[\tau_j, T]$ provides better evidence for an anomaly with intensity $\mu$ than $[\tau_i, T]$. As the difference between curves is unchanged as we observe more data, this in turn means that for any future point $T_F \geq T$, the interval $[\tau_j, T_F]$ will not provide better evidence than $[\tau_i, T_F]$. Therefore, the curve associated with $\tau_j$ can be pruned, removed from our computational characterisation of $S_T(\mu)$.

The following gives necessary and sufficient conditions for a curve to be dominated by another.

\begin{proposition}
Let $C^{(T)}_{\tau_i}$ and $C^{(T)}_{\tau_j}$ be curves that are positive somewhere on $\mu \in [1, \infty)$, where $\tau_i < \tau_j$ and $C_{\tau_i}^{(\tau_j-1)}$ is also positive somewhere on $\mu \in [1, \infty)$.

Then $C^{(T)}_{\tau_i}$ dominates $C^{(T)}_{\tau_j}$ if and only if $a_{\tau_j}^{(T)} / b_{\tau_j}^{(T)} \leq a_{\tau_i}^{(\tau_j-1)} / b_{\tau_i}^{(\tau_j-1)}$ or equivalently $a^{(T)}_{\tau_j} / b^{(T)}_{\tau_j} \leq a^{(T)}_{\tau_i} / b^{(T)}_{\tau_i}$. Additionally, it cannot be the case that $C_{\tau_j}^{(T)}$ dominates $C_{\tau_i}^{(T)}$.
\end{proposition}

A formal proof is given in Appendix \ref{sec:conditionsforpruning}, but we see intuitively why this result holds by looking at Figure \ref{fig:example_logarithm_curves}. We see that $C^{(T)}_{\tau_i}$ dominates $C^{(T)}_{\tau_j}$ precisely when $C^{(T)}_{\tau_i}$ has both a greater slope at $\mu=1$ (which occurs when $C_{\tau_i}^{(\tau_j-1)}$ is positive) and a greater root than $C^{(T)}_{\tau_j}$, where (as shown in the proof) the root of a curve $C_{\tau}^{(T)}$ is an increasing function of $a_{\tau}^{(T)}/b_{\tau}^{(T)}$.

The quantity $a_{\tau_i}^{(\tau_j-1)} / b_{\tau_i}^{(\tau_j-1)}$ does not change with time updates. Therefore, functional pruning occurs after periods of unpromising data where we have $x_t \approx \lambda$, bringing the ratios $a_{\tau}^{(T)}/b_{\tau}^{(T)}$ for curves down closer to $1$. This causes curves testing shorter more intense anomalies to be dominated by curves testing longer, less intense ones.

\section{Algorithm and theoretical evaluation}
\label{section:theoretical_evaluation}

Using this result we obtain the Poisson-FOCuS algorithm, described in Algorithm \ref{algorithm:focus_poisson}. This algorithm stores a list of curves in time order by storing their associated $a$ and $b$ parameters, as well as their times of creation $\tau$, which for the constant $\lambda$ case can be computed as $T+1-b/\lambda$.

On receiving a new observation at time $T$, these parameters are updated. If the observed count exceeds the expected count predicted by the most recent curve we also add a new curve which corresponds to a GRB that starts at time $T$. Otherwise we check to see if we can prune the most recent curve. This pruning step uses Proposition \ref{prop:bounded_curves}, which shows that if any currently stored curve can be pruned, the most recently stored curve will be able to be pruned. (Our pruning check does not does not need to be repeated for additional curves, as on average less than one curve is pruned at each timestep.)

The final part of the algorithm is to find the maximum of each curve, and check if the maximum of these is greater than the threshold. If it is, then we have detected a GRB. The start of the detected GRB is given by the time that the curve with the largest maximum value was added.

\begin{algorithm}[]
\SetAlgoLined
\KwResult{Startpoint, endpoint, and significance level of an interval above a $k$-sigma threshold.}
 initialise threshold $k^2/2$\; 
 initialise empty curve list\;
 \While{anomaly not yet found}{
 $ T \leftarrow T+1$\;
 get actual count $X_T$\;
 get expected count $\lambda$\;
 \tcp{update curves:}
 \For{curve $C_{\tau_i}^{(T-1)}$ in curve list $[C_{\tau_1}^{(T-1)}, ..., C_{\tau_n}^{(T-1)}]$}{
 $a_{\tau_i}^{(T)} \leftarrow a_{\tau_i}^{(T-1)}+X_T$;
 $b_{\tau_i}^{(T)} \leftarrow b_{\tau_i}^{(T-1)}+\lambda$;
 }
 \tcp{add or prune curve:}
  \uIf{$X_T/\lambda > \max[a_{\tau_n}^{(T)}/b_{\tau_n}^{(T)}, 1]$}{
    add $C_T^{(T)}: a_{T}^{(T)}=X_T, b_{T}^{(T)} = \lambda, \tau=T$ to curve list\;}
  \ElseIf{$a_{\tau_n}^{(T)}/b_{\tau_n}^{(T)} < \max[a_{\tau_{n-1}}^{(T)}/b_{\tau_{n-1}}^{(T)}, 1]$}{
    remove $C_{\tau_n}^{(T)}$ from curve list\;}
  
  \tcp{calculate maximum $M$:}
  \For{curve $C_{\tau_i}^{(T)}$ in curve list}{
    \If{$\max(C_{\tau_i}^{(T)}) > M$}{
    $M \leftarrow \max(C_{\tau_i}^{(T)})$\;
    $\tau^* \leftarrow \tau_i$
   }
   }
  \If{$M > k^2/2$}{
   anomaly found on interval $[\tau^*, T]$ with sigma significance $\sqrt{2M} > k$\;
   }
 }
 \caption{Poisson-FOCuS for constant $\lambda$}
 \label{algorithm:focus_poisson}
\end{algorithm}

\subsection{Dealing with varying background rate}
\label{section:varying_background}

Algorithm \ref{algorithm:focus_poisson} deals with the constant $\lambda$ case. If $\lambda = \lambda(t)$ is not constant, but an estimate of $\lambda(T)$ is available at each timestep $T$, we can apply the same principle but with a change in the definition of $b^{(T)}_{\tau}$. We now have $b^{(T)}_{\tau} := \sum_{t=\tau}^T \lambda(t)$, the total expected count over the interval $[\tau, T]$. For the algorithm, this impacts how the co-efficients are updated, with the new updates being
$$ a^{(T+1)}_{\tau} \leftarrow a^{(T)}_{\tau} + X_{T+1}, \ \ \ b^{(T+1)}_{\tau} \leftarrow b^{(T)}_{\tau} + \lambda(T+1) .$$
If we work with a non-homogeneous Poisson process in this way, it becomes impossible to recover $\tau$ from the coefficients $a^{(T)}_{\tau}$ and $b^{(T)}_{\tau}$, so $C_{\tau}^{(T)}$ must be computationally stored as the triplet $(\tau, a^{(T)}_{\tau}, b^{(T)}_{\tau})$.

The Poisson-FOCuS algorithm gives us an estimate of the start point of a GRB by reporting the interval $[\tau^*, T]$ over which an anomaly is identified. In our application, if the additional sanity checking indicates a GRB is present, the whole signal starting some time before $\tau^*$ is then recorded and transmitted from the spacecraft to Earth for a period of time. After this has occurred, Poisson-FOCuS can restart immediately provided that a good background rate estimate is available.

\begin{figure}
     \centering
     \begin{subfigure}[b]{0.49\textwidth}
         \centering
         \includegraphics[width=\textwidth]{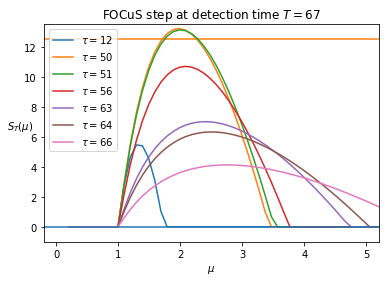}
     \end{subfigure}
     \hfill
     \begin{subfigure}[b]{0.49\textwidth}
         \centering
         \includegraphics[width=\textwidth]{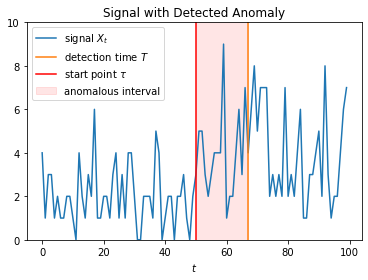}
     \end{subfigure}

     \caption{The FOCuS algorithm running on a random Poisson signal with background $\lambda=2$, with an anomaly of intensity $\mu \lambda = 4$ developing at start point $\tau=50$ and being detected at $T=67$.}
     \label{fig:poisson_focus}
\end{figure}

Figure \ref{fig:poisson_focus} shows an image of this algorithm running on a Poisson data signal. The Poisson-FOCuS algorithm correctly identifies the start point of the anomaly, as well as reporting a detection as soon as evidence crosses the $5$-sigma significance threshold.

\subsection{Minimum $\mu$ value}

For our application there is an upper limit on the length of a gamma ray burst. It thus makes sense to ensure we do not detect gamma ray bursts that are longer than this limit. 

To do so, we set an appropriate $\mu_{\text{min}}$, and additionally prune curves which only contribute to $S_T(\mu)$ on $1 < \mu < \mu_{\text{min}}$, by removing, or not adding, curves $C_{\tau}^{(T)}$ to the list if $C_{\tau}^{(T)}(\mu_{\text{min}}) \leq 0$, i.e.

\[ \frac{a^{(T)}_{\tau}}{b^{(T)}_{\tau}} \leq \frac{\mu_{\text{min}}-1}{\log \mu_{\text{min}}}.\]

\noindent We can choose $\mu_{\text{min}}$ according to our significance threshold and the maximum expected count we are interested in searching for bursts over, using the proof of Proposition \ref{prop:page_beats_window} about detectability for the window-based method, as follows:

\[ (h\lambda)_{\text{max}} = \frac{k^2}{2[\mu_{\text{min}} \log(\mu_{\text{min}}) - (\mu_{\text{min}}-1)]}.\]

\noindent For a $5$-sigma significance threshold, assuming a background rate of one photon every  $500\mu$s, a maximum length of 1 minute for a GRB would correspond to $\mu_{\text{min}}=1.015$, while a maximum length of 1 hour would have $\mu_{\text{min}}=1.002$.

\subsection{Using time-to-arrival data}
\label{section:other_data}

Rather than taking as data the number of photons observed in each time window, we can take as our data the time between each observation. In this case our data is $U_1,U_2,\ldots$ where $U_i$ is the time between the $(i+1)$th and $i$th photons. Under the assumption the data follows a Poisson process, we have that the $U_i$ are independently Exponentially distributed.

\begin{restatable}{proposition}{exponentialdata}
The Poisson-FOCuS algorithm still works in the Exponential case, with the only difference being how we update the co-efficients of the curves.

$$ a^{(T+1)}_{\tau} = a^{(T)}_{\tau} + 1, \ \ \ b^{(T+1)}_{\tau} = b^{(T)}_{\tau} + \lambda(T\!+\!1) U_{T+1},$$
where $\lambda(T)$ is the estimate of the backround rate at the time of the $T$th photon arrival.
\end{restatable}

\subsection{Computational cost comparisons}

Using a window method, our computational costs per window consist of: adding $x_T$ and $\lambda_T$ to the window; removing $x_{T-h}$ and $\lambda_{t-h}$ from the window; calculating the test statistic and comparing to the threshold. Using Poisson-FOCuS, our computational costs per curve are: adding $x_T$ to $a_{\tau}^{(T)}$; adding $\lambda_T$ to $b_{\tau}^{(T)}$; calculating the maximum of the curve and comparing to the threshold. The computational cost per curve is therefore roughly equal to the computational cost per window. Thus when evaluating the relative computational cost of Poisson-FOCuS versus a window method it is required to calculate the expected number of curves kept by the algorithm at each timestep, and compare against the number of windows used. We now give mathematical bounds on this quantity, as follows:

\begin{restatable}{proposition}{logcurves}
The expected number of curves kept by Poisson-FOCuS without $\mu_{\text{min}}$ at each timestep $T$ is $\in [\frac{\log(T)}{2}, \frac{\log(T)+1}{2}]$.
\end{restatable}

\begin{restatable}{proposition}{boundedcurves}
The expected number of curves kept by Poisson-FOCuS using some $\mu_{\text{min}} > 1$ at each timestep is bounded.
\label{prop:bounded_curves}
\end{restatable}

\begin{figure}
     \centering
     \begin{subfigure}[b]{0.49\textwidth}
         \centering
         \includegraphics[width=\textwidth]{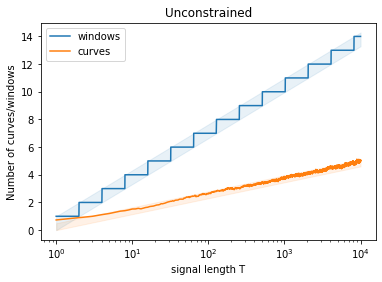}
         \caption{unconstrained case}
         \label{fig:unconstrained}
     \end{subfigure}
     \hfill
     \begin{subfigure}[b]{0.49\textwidth}
         \centering
         \includegraphics[width=\textwidth]{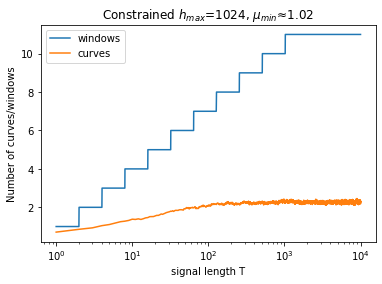}
         \caption{constrained case}
         \label{fig:constrained}
     \end{subfigure}
     
     \caption{Comparisons of the number of windows and expected number of curves (average over 1000 runs) kept by FOCuS running over a signal with base rate $\lambda=100$ using a $5$-sigma threshold. Figure \ref{fig:unconstrained} also highlights the logarithmic boundaries on the number of windows and curves in the unconstrained case.}
     \label{fig:cost_comparison}
\end{figure}

For geometrically spaced windows, over an infinite horizon the number of windows used at each timestep $T$ is $\in [\log_2(T), \log_2(T)+1]$, and if a $h_{\text{max}}$ is implemented then this will be bounded after a certain point. Figure \ref{fig:cost_comparison} gives a comparison of the number of windows and expected number of curves, showing that although the bound from Proposition \ref{prop:bounded_curves} is difficult to calculate, it is substantially below the corresponding bound on the number of windows. Therefore, Poisson-FOCuS provides the statistical advantages of an exhaustive window search at under half the computational cost of a geometrically spaced one.

\section{Empirical evaluation}
\label{section:empirical_evaluation}

We now empirically evaluate Poisson-FOCuS, and compare with a window method. We do this in two ways. First we simulate GRBs of different length and measure how many photon arrivals are needed within the GRB for them to be detected by each algorithm. Secondly, we run both methods on one week of data from the Fermi-GBM archive. 

Although not directly relevant to the improvements proposed by the Poisson-FOCuS algorithm, the issues with estimating the background rate and its impact on detecting GRBs are discussed in more detail in Appendix \ref{section:background_bias}.

\subsection{Statistical comparison with window method}

We first compare the FOCuS with a the window based method on synthetic data that has been simulated to mimic known GRBs, but allowing for different intensities of burst.
To simulate the data for a chosen known GRB at a range of different brightnesses, the photon stream of the GRB was converted into a random variable via density estimation. One draw from this random variable would give a photon impact time, and $n$ independent draws sorted into time order would give a stream of photon impact times that well approximate the shape of the burst. These were then overlaid on a background photon stream to form a signal that was binned into fundamental time widths of $50$ms, which was fed into either Poisson-FOCuS (equivalent to an exhaustive window search) or a geometrically spaced window search. The maximum sigma-level that was recorded when passing over the signal with each method is then plotted for various different brightnesses $n$. To stabilise any randomness introduced by the use of a random variable for GRB shape, this was repeated 10 times with different random seeds common to both methods and the average sigma-level is plotted.

The extent to which Poisson-FOCuS provides an improvement in detection power depends on the size and shape of the burst, and in particular whether the most promising interval in the burst lines up well with the geometrically spaced window grid. For example, the burst illustrated in Figure \ref{fig:grb1_pic} does not line up with this grid, and Figure \ref{fig:grb1_sim} shows how Poisson-FOCuS provides an improvement in detection power for this shape of burst at various different brightnesses. However, the shorter burst in Figure \ref{fig:grb2_pic} clearly has a most promising interval of size 1 for this binning choice, which is covered exactly by the geometric window grid. Therefore, Poisson-FOCuS provides no improvement over the window grid, as can be seen in Figure \ref{fig:grb2_sim}.

Because it is impossible to predict beforehand whether the most promising region of the burst will line up with any given choice of window grid, Poisson-FOCuS therefore provides a statistical advantage over any non-exhaustive choice of window grid in the real data setting.

\begin{figure}
     \centering
     \begin{subfigure}[b]{0.49\textwidth}
         \centering
         \includegraphics[width=\textwidth]{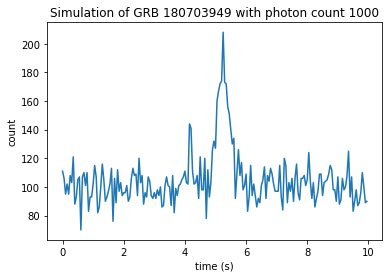}
         \caption{}
         \label{}
     \end{subfigure}
     \hfill
     \begin{subfigure}[b]{0.49\textwidth}
         \centering
         \includegraphics[width=\textwidth]{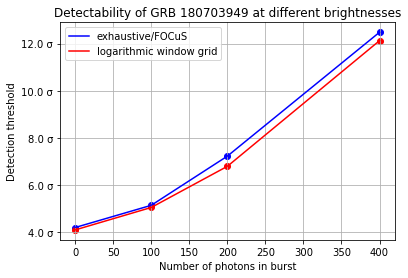}
         \caption{}
         \label{}
     \end{subfigure}
     \hfill
     
     \begin{subfigure}[b]{0.49\textwidth}
         \centering
         \includegraphics[width=\textwidth]{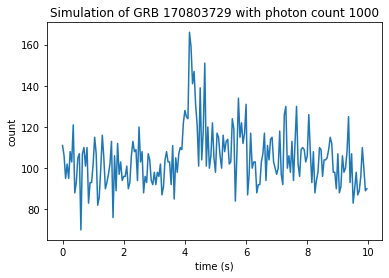}
         \caption{}
         \label{fig:grb1_pic}
     \end{subfigure}
     \hfill
     \begin{subfigure}[b]{0.49\textwidth}
         \centering
         \includegraphics[width=\textwidth]{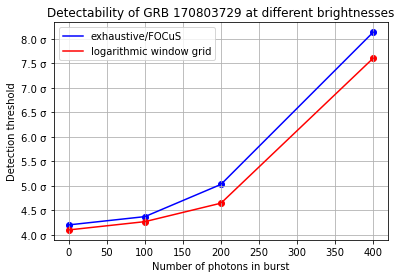}
         \caption{}
         \label{fig:grb1_sim}
     \end{subfigure}
     \hfill
     
     \begin{subfigure}[b]{0.49\textwidth}
         \centering
         \includegraphics[width=\textwidth]{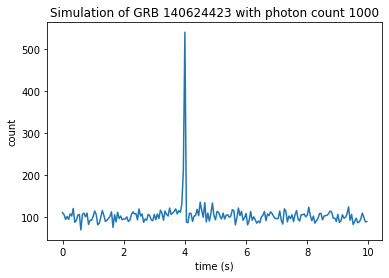}
         \caption{}
         \label{fig:grb2_pic}
     \end{subfigure}
     \hfill
     \begin{subfigure}[b]{0.49\textwidth}
         \centering
         \includegraphics[width=\textwidth]{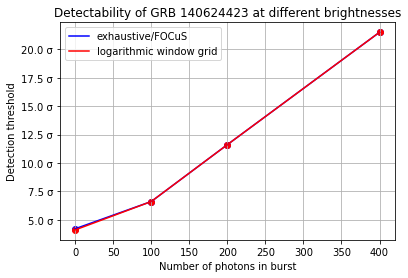}
         \caption{}
         \label{fig:grb2_sim}
     \end{subfigure}
     
     \caption{Plots of runs of FOCuS over simulated GRB copies of different brightnesses}
     \label{fig:simulated_grbs}
\end{figure}

\subsection{Application to FERMI data}
\label{section:FERMI_application}

In the context of HERMES, Poisson-FOCuS is currently being employed for two different purposes. First, a trigger algorithm built using Poisson-FOCuS is being developed for on-board, online GRB detection. To date, a dummy implementation has been developed and preliminary testings performed on the HERMES payload data handling unit computer. Second, Poisson-FOCuS is being employed in a software framework intended to serve as the foundation for the HERMES offline data analysis pipeline \cite[]{2022crupi}. In this framework, background reference estimates are provided by a neural network as a function of the satellite's current location and orientation.

Since no HERMES cube satellites have been launched yet, testing has taken place over Fermi gamma-ray burst monitor (GBM) archival data, looking for events which may have evaded the on-board trigger algorithm. The data used for the analysis were drawn from the Fermi GBM daily data, Fermi GBM trigger catalogue, and Ferm GBM untriggered burst candidates catalogue, all of which are publically available at NASA's High Energy Astrophysics Science Archive Research Center \cite[]{fermigdays, fermigtrig, fermiuntriggered}.

The algorithm was run over eight days of data, from 00:00:00 2017/10/02 to  23:59:59 2017/10/09 UTC time. This particular time frame was selected because, according to the untriggered GBM Short GRB candidates catalog, it hosts two highly reliable short GRB candidates which defied the Fermi-GBM online trigger algorithm. During this week the Fermi GBM algorithm was triggered by 11 different events. Six of these were classified as GRBs, three as terrestrial gamma-ray flashes, one as a local particle event and one as an uncertain event. The algorithm was run over data streams from 12 sodium iodide GBM detectors in the energy range of $50-300$ kiloelectron volts, which is most relevant to GRB detection but excludes the bismuth germanate detectors and higher energy ranges designed to find terrestrial gamma-ray flashes.

The data was binned at $100$ms. Background count-rates were assessed by exponential smoothing of past observations, excluding the most recent $4$s, and any curves corresponding to start points older than $4$s were automatically removed from the curve lists. The returning condition used was the same used by Fermi-GBM: a trigger is issued whenever at least two detectors are simultaneously above threshold. After a trigger, the algorithm was kept idle for five minutes and then restarted.

At a $5$-sigma threshold, Poisson-FOCuS was able to identify all the six GRBs which also triggered the Fermi-GBM algorithm, one of which is shown in Figures \ref{fig:burst_a} and \ref{fig:signif_b}. We also observed a trigger compatible with an event in the untriggered GBM Short GRB candidates catalog \cite[]{fermiuntriggered}, which is shown in Figures \ref{fig:burst_c} and \ref{fig:signif_d}. An uncertain event not in either catalogue is shown in Figures \ref{fig:burst_e} and \ref{fig:signif_f}, which may indicate a GRB that had been missed by earlier searches.

\afterpage{
\begin{figure}
\thispagestyle{empty}
     \centering
     \begin{subfigure}[b]{0.49\textwidth}
         \centering
         \includegraphics[width=\textwidth]{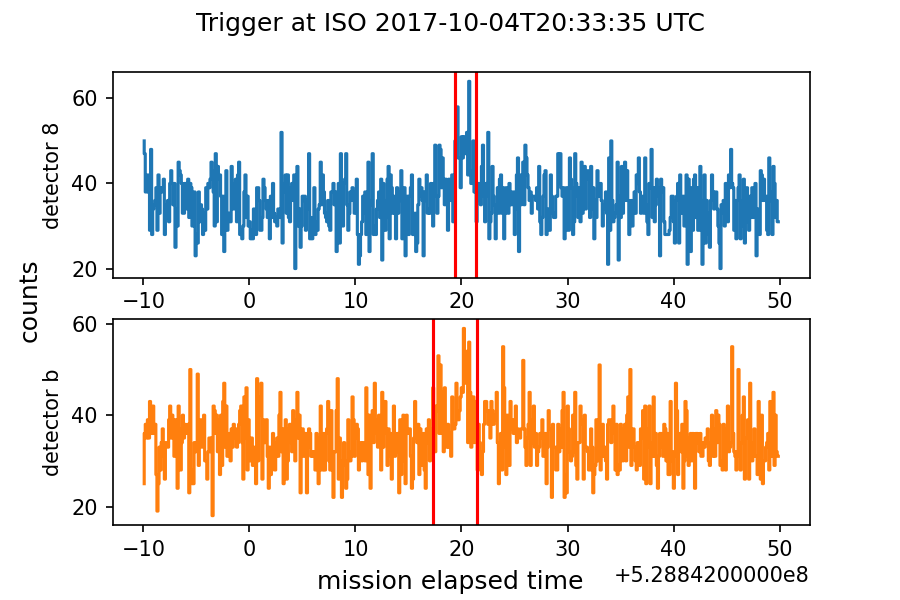}
         \caption{known GRB from Fermi catalogue}
         \label{fig:burst_a}
     \end{subfigure}
     \hfill
     \begin{subfigure}[b]{0.49\textwidth}
         \centering
         \includegraphics[width=\textwidth]{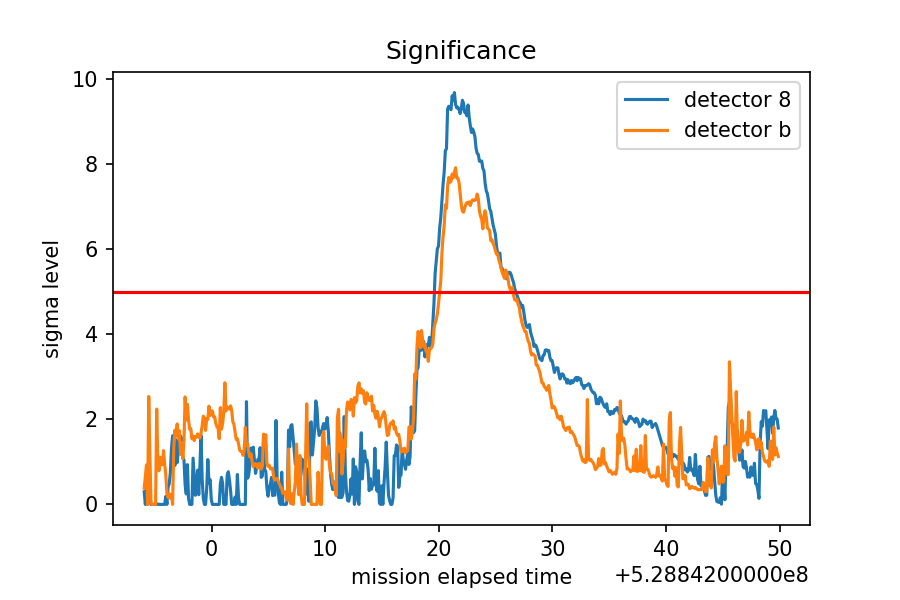}
         \caption{significance from Poisson-FOCuS}
         \label{fig:signif_b}
     \end{subfigure}      \hfill
     \begin{subfigure}[b]{0.49\textwidth}
         \centering
         \includegraphics[width=\textwidth]{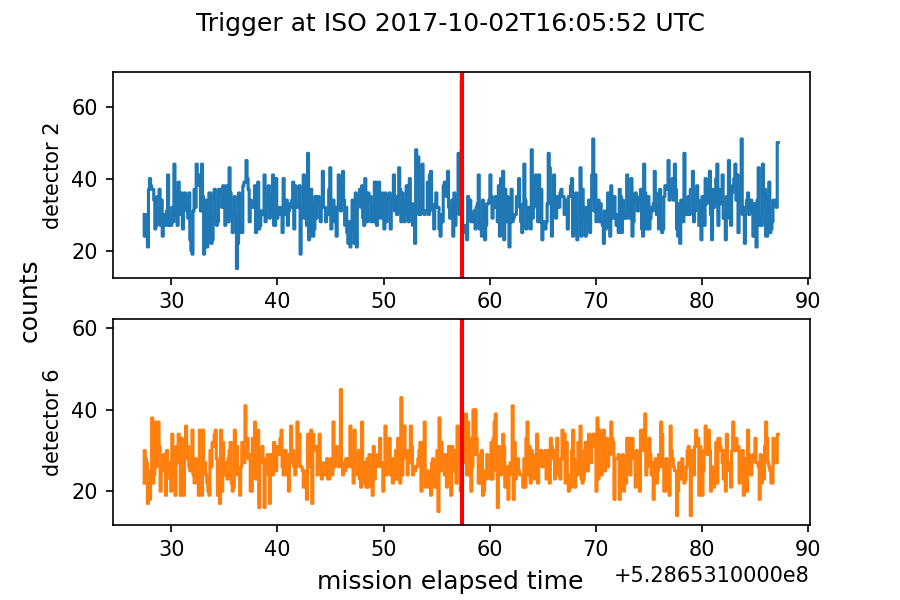}
         \caption{GRB in short candidates catalogue}
         \label{fig:burst_c}
     \end{subfigure}     \hfill
     \begin{subfigure}[b]{0.49\textwidth}
         \centering
         \includegraphics[width=\textwidth]{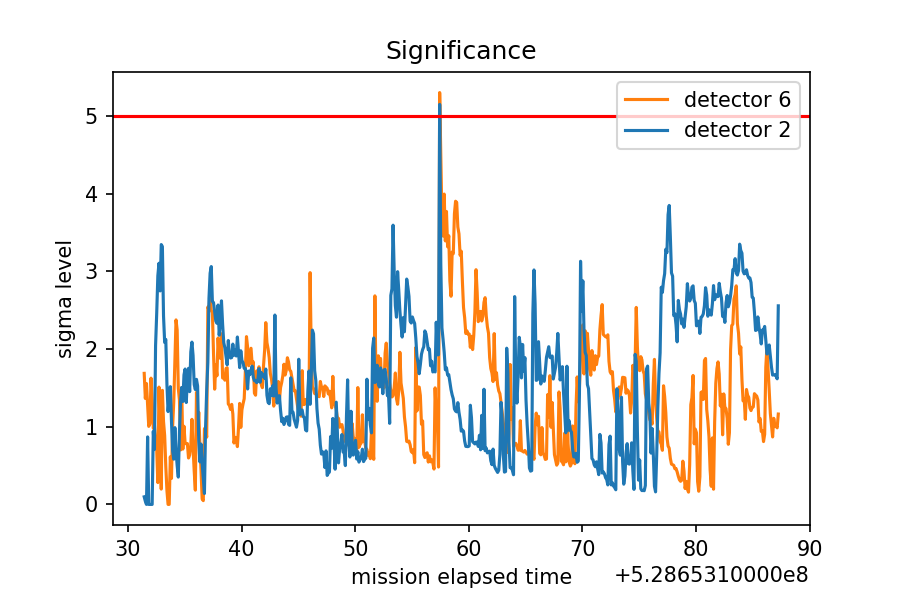}
         \caption{significance from Poisson-FOCuS}
         \label{fig:signif_d}
     \end{subfigure}
       \hfill
     \begin{subfigure}[b]{0.49\textwidth}
         \centering
         \includegraphics[width=\textwidth]{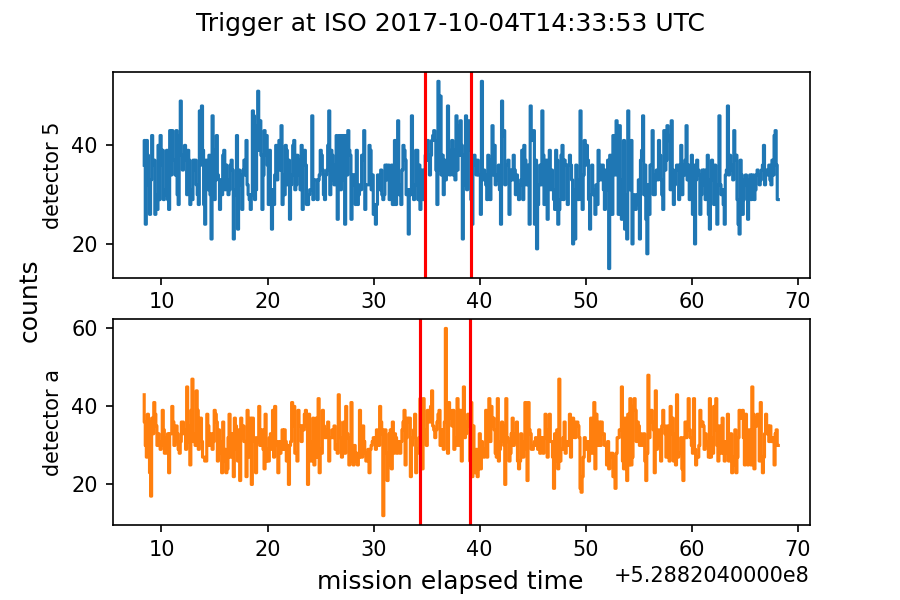}
         \caption{possible GRB not in either catalogue}
         \label{fig:burst_e}
     \end{subfigure}     \hfill
     \begin{subfigure}[b]{0.49\textwidth}
         \centering
         \includegraphics[width=\textwidth]{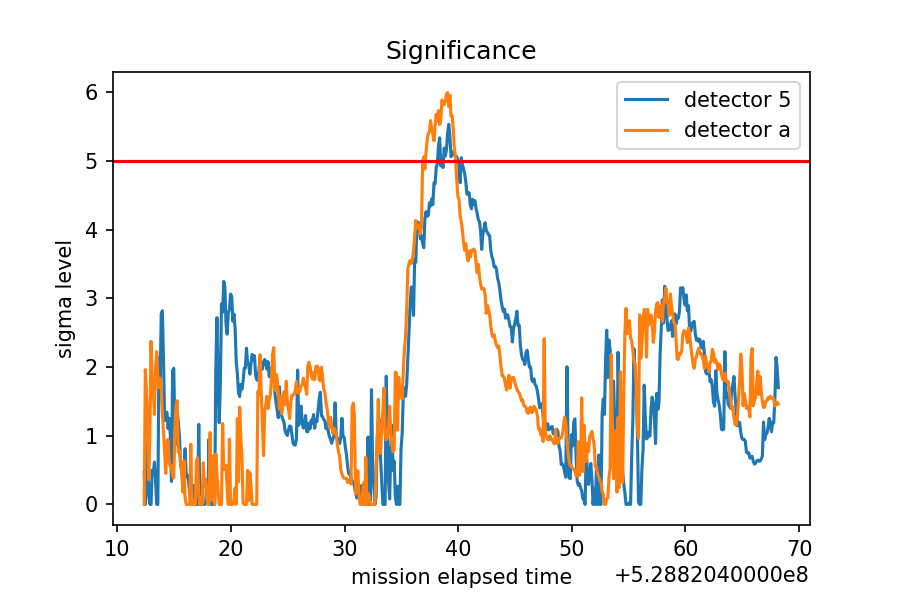}
         \caption{significance from Poisson-FOCuS}
         \label{fig:signif_f}
     \end{subfigure}  
     
     \caption{Three of the triggers found in the FERMI daily data. Left-hand column shows data from the two detectors that give a trigger, and the right-hand column shows the corresponding output from the Poisson-FOCuS algorithm.}
     \label{fig:fermi_application}
\end{figure}
\clearpage
}

\section{Discussion}
\label{section:discussion}


The main purpose of this work was to create a GRB detection algorithm that is mathematically equivalent to searching all possible window lengths, while requiring less computational power than the grid of windows approach. This was suitable for use on the HERMES satellites, where it has lead to a reduction of required computations in a very computationally constrained setting, as well as a reduction of number of parameter choices by practitioners as exact values for window lengths in a grid no longer need to be specified or justified.

There is increasing interest in detecting anomalies in other low-compute settings, for example Internet of Things sensors which must continuously monitor a signal \cite[]{Dey2018-dy}. These may have a limited battery life or limited electricity generation from sensor-mounted solar panels \cite[]{Nallusamy2011-iw}. Therefore, the algorithm we have developed may be of use more widely. 

Much of the mathematical work presented in this paper is also applicable to the $\mu \in [0, 1]$ case that searches for an anomalous lack of count in a signal. When adapting Poisson-FOCuS to this setting, it is important to make sure the algorithm functions well in situations where the counts are small, as these are precisely the locations of anomalies. This would likely entail using the adaption to work directly on count data given in Section \ref{section:other_data}, while ensuring that an anomalous lack of count could be declared in between individual counts. Combining these two cases would give a general algorithm for detection of anomalies on $\mu \in [0, \infty)$.

Code for Poisson-FOCuS and the analysis for this paper is available at the GitHub repository
\ifnotblind \url{https://github.com/kesward/FOCuS} \else \url{https://github.com/***} \fi.

\ifnotblind
\noindent
{\bf Acknowledgements}

This work was supported by the EPSRC grants EP/N031938/1 and EP/R004935/1, and BT as part of the Next Generation Converged Digital Infrastructure (NG-CDI) Prosperity Partnership.

\fi

\bibliographystyle{agsm}
\bibliography{bibliography, biblio} 

\newpage
\ifJASA
\setcounter{page}{1}
{\bf Supplementary Material for ``Poisson-FOCuS: An efficient online method for detecting count bursts with application to gamma ray burst detection"}
\fi

\appendix 

\section{Derivations of the LR statistic}

\subsection{Window method}

\lrderivation*

\begin{proof}
On an interval $x_{t+1:t+h}$, we have expected count $h\lambda$ and actual count $ h\bar{x}_{t+1:t+h}$. We utilise the Poisson likelihood

\[ L(\lambda; x_{t+1:t+h}) =  \frac{e^{-h\lambda} (h\lambda)^{h\bar{x}_{t+1:t+h}}}{(h\bar{x}_{t+1:t+h})!},\]

and log-likelihood

\[ \ell(\lambda; x_{t+1:t+h}) = -h\lambda + h\bar{x}_{t+1:t+h} \log (h\lambda) + c. \]

Our likelihood ratio statistic then becomes

\[
LR= -2 \left\{ \ell(\lambda; x_{t+1:t+h}) - \ell(\bar{x}_{t+1:t+h}; x_{t+1:t+h}) \right\} \]

\[ =  - 2 \left\{ -h\lambda + h\bar{x}_{t+1:t+h} \log (h\lambda) - (-h\bar{x}_{t+1:t+h} + h\bar{x}_{t+1:t+h} \log (h\bar{x}_{t+1:t+h}) )\right\} \]

\[ =  2h\lambda \left\{\frac{\bar{x}_{t+1:t+h}}{\lambda} \log \left(\frac{\bar{x}_{t+1:t+h}}{\lambda}\right) -  \left(\frac{\bar{x}_{t+1:t+h}}{\lambda}-1\right) \right\}. \]

\end{proof}

\subsection{Page-CUSUM method}

\lrderivationpage*

\begin{proof}
Our Poisson likelihood and log-likelihood is as follows:

$$L(\lambda;x_{1:T} ) = \frac{e^{-T\lambda}(T\lambda)^{\sum_{t=1}^T x_t}}{(\sum_{t=1}^T x_t)!},$$

$$l(\lambda;x_{1:T}) = -T\lambda + \sum_{t=1}^T x_t\log(\lambda) + c.$$

Under the null hypothesis of no anomaly, and the alternative of one anomaly at $\tau$, we have as our log-likelihoods the following:

$$l(\mathbf{H}_0;x_{1:T} ) = \sum_{t=1}^{T}[x_t \log(\lambda) - \lambda] + c$$

$$l(\mathbf{H}_1;x_{1:T}) = \max_{1 \leq \tau \leq T} \left(\sum_{t=1}^{\tau-1}[x_t \log(\lambda) - \lambda]+ \sum_{t=\tau}^{T}[x_t \log(\mu \lambda) - \mu \lambda]\right) + c$$

Here, the maximum is because we have no idea where our start point $\tau$ actually is, so we look at them all and pick the one with largest likelihood. This gives our log-likelihood ratio statistic as

$$LR = 2\max_{1 \leq \tau \leq T} \sum_{t=\tau}^T(x_t \log (\mu) - \lambda (\mu-1)) $$

$$= \max_{1 \leq \tau \leq T} \left[2(T-\tau+1)\lambda \left\{\frac{\bar{x}_{\tau:T}}{\lambda} \log \left(\mu \right) -  \left(\mu-1\right) \right\} \right].$$
\end{proof}

\subsection{Exponential}

\exponentialdata*

\begin{proof}
Making the assumption that we can consider the background rate constant between successive photon arrivals, our hypotheses for an individual logarithm curve $C_{\tau}^{(T)}$ are as follows:

\begin{itemize}
    \item $\mathbf{H}_0$: $U_{\tau},\ldots, U_{T}$ has $U_t \sim \text{Exp}(\lambda_t)$.
    \item $\mathbf{H}_1$: $ U_{\tau},\ldots,U_{T}$ has $U_t \sim \text{Exp}(\mu \lambda_t)$, for some $\mu>1$.
\end{itemize}

The exponential likelihood and log-likelihood are as follows:

\[ L(\lambda_{\tau:T}; u_{\tau:T}) =  \prod_{t=\tau}^T (\lambda_t) e^{-\sum_{t=\tau}^T \lambda_t u_t}, \]

\[ l(\lambda_{\tau:T}; u_{\tau:T}) =  \sum_{t=\tau}^T \log(\lambda_t) - \sum_{t=\tau}^T \lambda_t u_t.\]

This gives our log-likelihood ratio for the curve as

\[ C_{\tau}^{(T)} := l(\mu \lambda_{\tau:T}; u_{\tau:T}) - l(\lambda_{\tau:T}; u_{\tau:T}) =  \sum_{t=\tau}^T [\log(\mu) - \lambda_t u_t(\mu-1)] \]

This gives the update coefficients for curves as stated.

\end{proof}

\section{Detectability regions}

Here we provide the derivations of the detectability regions. For ease of reference, we reproduce these regions in Figure \ref{fig:page_window_comparison2}. 

Assume we are running a window of length $h$ over a signal containing a burst $x_{t+1:t+h^*}$ of length $h^*$. Our background rate $\lambda$ is assumed fixed. We want to figure out what is the smallest intensity $\mu^*$ we are able to detect, assuming that $\mu$ is the faintest intensity at which a burst of length $h$ is detectable.

Bursts of duration $h^*>h$ will only be detected at the $k$-sigma significance level if some subinterval of size $h$ is detected at the $k$-sigma significance level. No additional benefit can be provided by the presence of the part of the burst currently outside the window, so $\mu^*=\mu$. Therefore the green line on Figure \ref{fig:page_window_comparison2} has been drawn as a straight vertical.

Bursts of a duration $h^*<h$ can be found if they have a higher $\mu^*$. Splitting the window $h$ into anomalous and non-anomalous parts, we have that

$$ \mu \lambda h = \mu^* \lambda h^* + \lambda (h - h^*).$$

This rearranges to

$$ (\mu-1) h = (\mu^*-1) h^*,$$

which gives the other green line shown in Figure \ref{fig:page_window_comparison2}.

\begin{figure}
     \centering
     \begin{subfigure}[b]{0.6\textwidth}
         \centering
         \includegraphics[width=\textwidth]{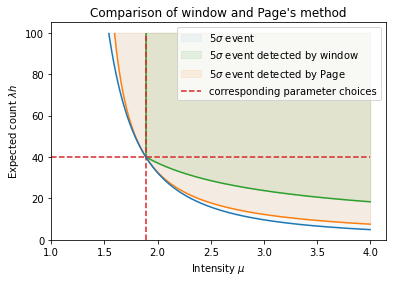}
     \end{subfigure}
     \caption{Detectability of Page-CUSUM and Window methods. \label{fig:page_window_comparison2}}
\end{figure}

Assume we are using Page's method with parameter $\mu$ over a signal containing a burst of intensity  $\mu^*$. Our background rate $\lambda$ is assumed fixed. We want to figure out what is the shortest duration $h^*$ required to detect the burst at a $k$-sigma threshold. Using our likelihood ratio, we have that:

\[ h^*\lambda \left[ \mu^* \log \left(\mu\right) -  \left(\mu-1\right) \right] = \frac{k^2}{2}, \]

\[  \mu^* = \frac{1}{\log(\mu)}\left[ \frac{k^2}{2h^* \lambda} + (\mu-1) \right] .\]

This gives the orange line in Figure \ref{fig:page_window_comparison}.

\section{Equivalences between Page-CUSUM and window methods}

\pageonlywindow*

\begin{proof}
Consider the last time $\tau$ where Page's statistic last became non-zero. On $[\tau, T]$ the likelihood ratio with our choice of $\mu>1$ exceeds a $k$-sigma threshold, therefore the maximised likelihood ratio over an unconstrained $\mu>1$ (which occurs at $\mu = \frac{\bar{x}_{\tau:T}}{\lambda}$) also exceeds a $k$-sigma threshold.
\end{proof}

\pagebeatswindow*

\begin{proof}

We choose the value of $\mu$ solving the equation

\[ 2h\lambda \left[ \mu \log \left(\mu\right) - \left(\mu-1\right) \right] = k^2, \]

i.e. the ideal intensity choice for the expected count $h\lambda$ used in this likelihood ratio test. Since $x_{t+1:t+h}$ is significant at the $k$-sigma level, we have that

\[2h\lambda \left[\frac{\bar{x}_{t+1:t+h}}{\lambda} \log \left(\frac{\bar{x}_{t+1:t+h}}{\lambda}\right) -  \left(\frac{\bar{x}_{t+1:t+h}}{\lambda}-1\right) \right] \geq k^2. \]

As the function $f(x) = x \log x - (x-1)$ is an increasing function, this shows that $\bar{x}_{t+1:t+h}/\lambda \geq \mu$.

We then have that

\begin{align*}
    S_{t+h}(\mu) &= \left[\max_{1 \leq \tau \leq t+h} \sum_{s=\tau}^{t+h}(x_s \log (\mu) - \lambda (\mu-1))\right]^+ \\
            &\geq \sum_{s=t+1}^{t+h} (x_s \log (\mu) - \lambda (\mu-1)) \\
            &= h\lambda \left[ \frac{\bar{x}_{t+1:t+h}}{\lambda} \log (\mu) -  (\mu-1) \right] \\
            &\geq  h\lambda \left[ \mu \log \left(\mu\right) -  \left(\mu-1\right) \right] \\
            &= \frac{k^2}{2}.
\end{align*} 

Therefore $S_{t+h}(\mu)$ is significant at a $k$-sigma significance level.

\end{proof}

\section{Conditions for pruning} \label{sec:conditionsforpruning}

\begin{proposition}
Let $C^{(T)}_{\tau_i}$ and $C^{(T)}_{\tau_j}$ be curves that are positive somewhere on $\mu \in [1, \infty)$, where $\tau_i < \tau_j$ and $C_{\tau_i}^{(\tau_j-1)}$ is also positive somewhere on $\mu \in [1, \infty)$.

Then $C^{(T)}_{\tau_i}$ dominates $C^{(T)}_{\tau_j}$ if and only if $a_{\tau_j}^{(T)} / b_{\tau_j}^{(T)} \leq a_{\tau_i}^{(\tau_j-1)} / b_{\tau_i}^{(\tau_j-1)}$ or equivalently $a^{(T)}_{\tau_j} / b^{(T)}_{\tau_j} \leq a^{(T)}_{\tau_i} / b^{(T)}_{\tau_i}$. Additionally, it cannot be the case that $C_{\tau_j}^{(T)}$ dominates $C_{\tau_i}^{(T)}$.
\end{proposition}
\begin{proof}
Let $\mu_{ij}$ be the non-unit intersection point of $C^{(T)}_{\tau_i}$ and $C^{(T)}_{\tau_j}$, i.e. the root of $C^{(\tau_j-1)}_{\tau_i}$. Then by rearrangement we have that

$$ a_{\tau_i}^{(\tau_j-1)}\log(\mu_{ij}) - b_{\tau_i}^{(\tau_j-1)}(\mu_{ij}-1) = 0,$$

$$ \frac{a_{\tau_i}^{(\tau_j-1)}}{b_{\tau_i}^{(\tau_j-1)}} = \frac{\mu_{ij}-1}{\log(\mu_{ij})}.$$

Because $C_{\tau_i}^{(\tau_j-1)}$ is non-negative on $\mu \in [1, \mu_{ij})$, we cannot have $C_{\tau_j}^{(T)}$ dominating $C_{\tau_i}^{(T)}$. For $C_{\tau_i}^{(T)}$ to dominate $C_{\tau_j}^{(T)}$, we must have that $C_{\tau_j}^{(T)} \leq 0$ on $\mu \in [\mu_{ij}, \infty)$, i.e. $C_{\tau_j}^{(T)}(\mu_{ij}) \leq 0$. Rearranging, we have

$$ a_{\tau_j}^{(T)}\log(\mu_{ij}) - b_{\tau_j}^{(T)}(\mu_{ij}-1) \leq 0,$$

$$ \frac{a_{\tau_j}^{(T)}}{b_{\tau_j}^{(T)}} \leq \frac{\mu_{ij}-1}{\log(\mu_{ij})}.$$

Putting these together gives us the condition $a_{\tau_j}^{(T)} / b_{\tau_j}^{(T)} \leq a_{\tau_i}^{(\tau_j-1)} / b_{\tau_i}^{(\tau_j-1)}$. For the other form, note that we can rearrange the inequality:

\[ a^{(T)}_{\tau_j} b^{(\tau_j-1)}_{\tau_i}  \leq a^{(\tau_j-1)}_{\tau_i} b^{(T)}_{\tau_j},  \]

\[ a^{(T)}_{\tau_j} b^{(\tau_j-1)}_{\tau_i} +  a^{(T)}_{\tau_j} b^{(T)}_{\tau_j} \leq a^{(\tau_j-1)}_{\tau_i} b^{(T)}_{\tau_j} +  a^{(T)}_{\tau_j} b^{(T)}_{\tau_j},  \]

\[ a^{(T)}_{\tau_j} b^{(T)}_{\tau_i}  \leq a^{(T)}_{\tau_i} b^{(T)}_{\tau_j}.\]

\end{proof}

\section{Bounds on numbers of curves}

\logcurves*
\begin{proof}

Recalling that a logarithm curve $C_{\tau}^{(T)}(\mu)$ is defined as

\[ C_{\tau}^{(T)}(\mu) := \sum_{t=\tau}^T[X_t \log(\mu)-\lambda(\mu-1)],\]

we define the set of candidate start points $\mathfrak{I}_{T}$ at time $T$ to be the set of all $\tau$ directly contributing to $S_T(\mu)$, i.e.

\[ \mathfrak{I}_{T} := \{ \tau: \exists \mu, \forall \tau^{'} \neq \tau,  [C_{\tau}^{(T)}(\mu)]^+ > [C_{\tau^{'}}^{(T)}(\mu)]^+ \} .\]

The number of curves kept by Poisson-FOCuS at time $T$ is, barring computational implementations that occasionally keep extra curves to avoid repeated pruning checks, exactly $|\mathfrak{I}_{T}|$.

\begin{lemma}
Suppose $\tau^{'} \in \mathfrak{I}_{T}$. This is equivalent to the following two conditions:

\begin{itemize}
    \item for any $\tau^{'} < \tau^{''} \leq T$, we have that 
            $$\lambda < \bar{X}_{\tau^{'}:\tau^{''}} .$$
    \item for any $1 \leq \tau < \tau^{'} < \tau^{''} \leq T$, we have that
            \[ \bar{X}_{\tau, \tau^{'}-1} < \bar{X}_{\tau^{'}, \tau^{''}}. \]
\end{itemize}

\end{lemma}

\begin{proof}
Suppose $\exists  \tau, \tau^{''}$ such that we have 

\[ \bar{X}_{\tau:\tau^{'}-1} \geq \bar{X}_{\tau^{'}:\tau^{''}}. \]

Consider the two curves $C_{\tau}^{(\tau^{''})}(\mu)$ and $C_{\tau^{'}}^{(\tau^{''})}(\mu)$:

\begin{align}
  C_{\tau}^{(\tau^{''})}(\mu) &= \sum_{t=\tau}^{\tau^{''}}[X_t \log(\mu)-\lambda(\mu-1)] \\
                              &= [(\tau^{'}-\tau)\bar{X}_{\tau:\tau^{'}-1} + (\tau^{''}-\tau^{'}+1)\bar{X}_{\tau^{'}:\tau^{''}}]\log(\mu) - [\tau^{''}-\tau+1]\lambda(\mu-1) \\
                              &\geq [(\tau^{''}-\tau+1)\bar{X}_{\tau:\tau^{''}}]\log(\mu) - [\tau^{''}-\tau+1]\lambda(\mu-1) \\
                              &= \frac{\tau^{''}-\tau+1}{\tau^{''}-\tau^{'}+1} C_{\tau^{'}}^{(\tau^{''})}(\mu) \\
                              &\geq C_{\tau^{'}}^{(\tau^{''})}(\mu).
\end{align}

So $\tau^{'} \not\in \mathfrak{I}_{T}$.

What this is saying is that if $\bar{X}_{\tau, \tau^{'}-1} < \bar{X}_{\tau^{'}, \tau^{''}}$, then the interval $[\tau, \tau^{''}]$ has both greater intensity and greater duration than the interval $[\tau^{'}, \tau^{''}]$, so $\tau^{'}$ cannot be a candidate start point.

To prove the reverse, we note that the non-unit point of intersection between $C_{\tau^{'}}^{(\tau^{''})}(\mu)$ and $0$ is a monotone increasing function of $\bar{X}_{\tau^{'}:\tau^{''}}$. Therefore, if for all $\tau < \tau^{'} < \tau^{''} \leq T$,

$$\bar{X}_{\tau:\tau^{''}} < \bar{X}_{\tau^{'}:\tau^{''}},\ \ \ \lambda < \bar{X}_{\tau^{'}:\tau^{''}},$$

we must have that $\exists \mu > 1$ such that $[C_{\tau^{'}}^{(\tau{''})}(\mu)]^+ > [C_{\tau}^{(\tau{''})}(\mu)]^+$. This gives $\tau^{'} \in \mathfrak{I}_{T}$.

\end{proof}

\begin{lemma}\label{lem:convex_minorant}
Define the sequence $Z_0 := 0$, $Z_T := \sum_{t=1}^T X_t$.

If $\tau$ in $\mathfrak{I}_{T}$, then:

\begin{itemize}
    \item $\tau-1$ is an extreme point of the largest convex minorant of the sequence $\{ Z_t-t\lambda: t \leq T \}$.
    \item $\forall T \geq t > \tau-1$, we additionally have that $Z_t-t\lambda > Z_{\tau-1}-(\tau-1)\lambda$, i.e. $\tau-1$ is on the "right-hand side" of the convex minorant.
\end{itemize}
\end{lemma}

\begin{proof}

Let $\tau$ in $\mathfrak{I}_{T}$.

As above, we have that for any $1 \leq \tau < \tau^{'} < \tau^{''} \leq T$:

\[ \bar{X}_{\tau, \tau^{'}-1} < \bar{X}_{\tau^{'}, \tau^{''}}. \]

This can be equivalently written as

\[ \frac{Z_{\tau^{'}-1} - Z_{\tau}}{(\tau^{'}-1)-\tau}  <  \frac{Z_{\tau^{''}} - Z_{\tau^{'}-1}}{\tau^{''}-(\tau^{'}-1)},  \]

which shows that $\tau-1$ is in the largest convex minorant of $Z_t$, and therefore of $Z_t-t\lambda$.

To show we are on the right-hand side of this convex minorant, we assume hoping for a contradiction that $\exists \tau^{'} \geq \tau$ such that $Z_{\tau^{'}}-\tau^{'}\lambda < Z_{\tau-1}-(\tau-1)\lambda$. We then have that

\[ C_{\tau}^{(\tau^{'})}(\mu) = \sum_{t=\tau}^{\tau^{'}}[X_t \log(\mu)-\lambda(\mu-1)] \]

\[ = [Z_{\tau^{'}}-Z_{\tau-1}]\log(\mu) - [\tau^{'}-(\tau-1)]\lambda(\mu-1) \]

\[ < [\tau^{'}-(\tau-1)]\lambda[\log(\mu) - (\mu-1)] \]

\[ < 0.\]

So for $\tau^{'}=\tau$ we have that $X_\tau < \lambda$, and for $\tau^{'}>\tau$ we have that $C_{\tau}^{(T)}(\mu) < C_{\tau^{'}}^{(T)}(\mu)$ pointwise. Either way, $\tau \not\in \mathfrak{I}_{T}$.

To prove the reverse, note that the argument for being on the convex minorant is entirely reversible, and that $Z_t-t\lambda < Z_{\tau-1}-(\tau-1)\lambda$ is equivalent to $\bar{X}_{\tau, t} < \lambda$.

\end{proof}

\begin{figure}[hbt!]
    \centering
    \includegraphics[width=0.49\textwidth]{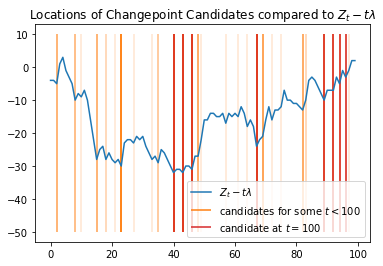}
    \caption{Plots of $\tau-1$ for each anomaly start point $\tau$ compared to the random walk $Z_t-t\lambda$.}
    \label{random_walk}
\end{figure}

Figure \ref{random_walk} shows what this looks like. For a signal $X_t$ that Poisson-FOCuS is run over, the random walk $Z_t-t\lambda$ is plotted. Values of $\tau-1$ for each candidate anomaly start point $\tau$ are highlighted in orange, with the intensity of the highlight corresponding to how long they were kept, or highlighted in red if they were still kept by the time $T=100$.

For each candidate point $\tau$ with $\tau-1$ highlighted in red, we have that:
\begin{itemize}
    \item Convex minorant: The gradient drawn through $Z_{\tau-1}-(\tau-1)\lambda$ from any point before $\tau-1$ must be less than the gradient drawn through $Z_{\tau-1}-(\tau-1)\lambda$ by any point after $\tau-1$.
    \item Right side: It is possible to draw a straight horizontal line from $Z_{\tau-1}-(\tau-1)\lambda$ to the right side of the graph without crossing any other $Z_t-t\lambda$.
\end{itemize}

This is approximately half the points in the convex minorant of $Z_t-t\lambda$ (the other half being the left-hand side, with one point - the minimum - being in both).

\begin{theorem}\label{thm:andersen}

Let $X_t, 1 \leq t \leq T$ be independent identically distributed continuous random variables, and let $S_t := \sum_{s=1}^t X_s$ be the corresponding random walk. Then the number of points $H(T)$ on the convex minorant of the sequence $(0, S_1, S_2, ..., S_T)$ (not including endpoints) has the distribution

\[ H(T) \sim \sum_{t=1}^{T-1} Y_t,  \]

\[ Y_t \sim \text{Bernoulli}\left(\frac{1}{t+1}\right), \]

where the $Y_t$ are independent of each other and the distribution of the $X_t$.

\end{theorem}

\begin{proof}
See \cite{Andersen1954-ik}
\end{proof}

Under the null hypothesis, we have $X_t \sim \text{Poisson}(\lambda)$ independent and identically distributed.

By Lemma \ref{lem:convex_minorant}, the $|\mathfrak{I}_{T}|$ is the number of points on the right-hand side of the convex minorant of $ \{ Z_t-t\lambda: t \leq T\} $.

By Therorem \ref{thm:andersen}, the expected number of points on the convex minorant of $ \{ Z_t-t\lambda: t \leq T\} $ not including endpoints is 
\[ \mathbb{E}[\text{points on convex minorant}] = \sum_{t=2}^{T} \frac{1}{t}. \]

By symmetry, the expected number of points on the right-hand side of the convex minorant (including the minimum point, which is on both sides), is

\[ \mathbb{E}[|\mathfrak{I}_{T}|] = \frac{1}{2}\sum_{t=1}^{T} \frac{1}{t}. \]

Because $\tau \in \mathfrak{I}_T$ is related to $\tau-1$ (rather than $\tau$) being on the right-hand side of the convex minorant of $\{ Z_t-t\lambda: t \leq T\}$, it is impossible to be on the rightmost endpoint. However, it could be that $\tau-1 = 0$ could be both the leftmost endpoint and on the right-hand side of the convex minorant, which would give an additional curve beyond those given by Theorem \ref{thm:andersen}. This would require $\min\{ Z_t-t\lambda: t \leq T\} \geq 0$, which has a probability that $\rightarrow 0$ as $T \rightarrow \infty$, and can therefore be discounted for large values of $T$ as it falls within the harmonic upper bound given below.

We have by standard results for harmonic sums that

\[ \sum_{t=1}^{T} \frac{1}{t} \in [\log(T), \log(T)+1],  \]

Giving us that

\[ \mathbb{E}[|\mathfrak{I}_{T}|] \in \left[\frac{\log(T)}{2}, \frac{\log(T)+1}{2}\right]. \]

\end{proof}

\boundedcurves*

\begin{proof}

Let $\lambda > 0$, $\mu_{\text{min}} > 1$ be fixed, and $X_T \sim \text{Poisson}(\lambda)$. Define $S_0 = 0$, and for each $T \in \mathbb{N}$ recursively define

\[  S_{T+1} = S_T + X_{T+1}\log(\mu_{\text{min}}) - \lambda(\mu_{\text{min}}-1).  \]

This gives essentially Page's statistic without resetting negative values to zero. We further define:

\[ H(X) := \inf_{T>1} \{ T: S_T \leq 0 \}  \]

i.e. $H(X)$ is the time elapsed between resets to zero of Page's statistic.

In order to prove positive recurrence of Page's statistic, we now show that $\mathbb{E}[H(X)]$ is finite.

We have that

\[ \mathbb{E}[S_T] = \lambda T [\log(\mu_{\text{min}}) - (\mu_{\text{min}}-1)] < 0, \]
\[ \text{Var}[S_T] = \lambda T (\log(\mu_{\text{min}}))^2 < \infty.\]

This gives that $H(X) < \infty$ almost surely.

By the central limit theorem, we have that as $T \rightarrow \infty$,

\[ \frac{S_T - \lambda T [\log(\mu_{\text{min}}) - (\mu_{\text{min}}-1)]}{\sqrt{\lambda T}\log(\mu_{\text{min}})} \approx N(0, 1).  \]

Using this approximation we can then calculate

\[ S_T \approx N(0, 1) \sqrt{\lambda T}\log(\mu_{\text{min}}) + \lambda T [\log(\mu_{\text{min}}) - (\mu_{\text{min}}-1)].  \]

\[  \mathbb{P}(S_T < 0) \approx \Phi\left(\sqrt{\lambda T}\left[1 - \frac{(\mu_{\text{min}}-1)}{\log(\mu_{\text{min}})}\right]\right) \]

This gives us the following bound:

\begin{align}
    \mathbb{E}[H(X)] &= \sum_{T=1}^{\infty} \mathbb{P}[H(X) \leq T] \\
                     &\leq \sum_{T=1}^{\infty} \mathbb{P}(S_T < 0) \\
                     &\approx \sum_{T=1}^{\infty} \Phi\left(\sqrt{\lambda T}\left[1 - \frac{(\mu_{\text{min}}-1)}{\log(\mu_{\text{min}})}\right]\right) \\
                     &< \infty.
\end{align}

The last step is because the Gaussian distribution has tails that drop as the square of an exponential, and geometric series have finite sum.

Therefore, the expected number of curves in the FOCuS algorithm running using a $\mu_{\text{min}}$ is bounded, because all curves in the algorithm are removed each time Page's statistic using $\mu_{\text{min}}$ resets to $0$, and the expected time between resets is finite.

\end{proof}

\section{Bias from estimating background rate}
\label{section:background_bias}

When using Poisson-FOCuS on a dataset requiring background rate estimation, particular care needs to be taken with the choice of the estimator for the background rate. This is because:

\begin{itemize}
    \item The presence of an anomaly within the data being used to estimate background rate could destabilise the estimation, so robust methods are preferred.
    \item The ability of Poisson-FOCuS to give immediate detections requires a background estimate for time $T$ to be available using only data from $t \leq T$. Delaying this estimate will also delay detections.
    \item Small, consistent biases in background estimation will be recorded as anomalies over long timescales.
\end{itemize}

\begin{figure}[hbt!]
    \centering
    \includegraphics[width=0.49\textwidth]{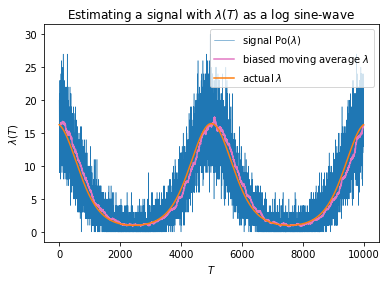}
    \caption{A log-sine wave using a biased method to recover $\lambda(T)$}
    \label{fig:signal_estimation}
\end{figure}

\begin{figure}
     \centering
     \begin{subfigure}[b]{0.49\textwidth}
         \centering
         \includegraphics[width=\textwidth]{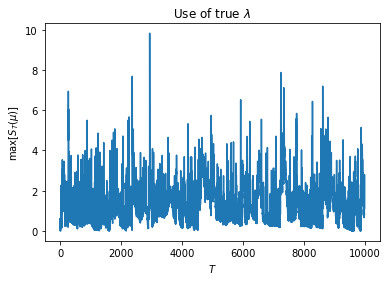}
         \caption{}
         \label{fig:no_estimation}
     \end{subfigure}
     \hfill
     \begin{subfigure}[b]{0.49\textwidth}
         \centering
         \includegraphics[width=\textwidth]{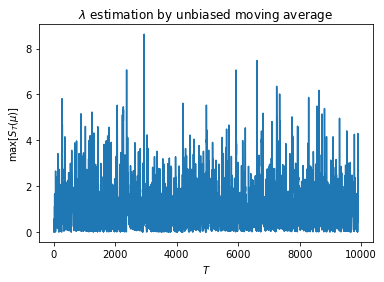}
         \caption{}
         \label{fig:unbiased_estimation}
     \end{subfigure}
     \hfill
     \begin{subfigure}[b]{0.49\textwidth}
         \centering
         \includegraphics[width=\textwidth]{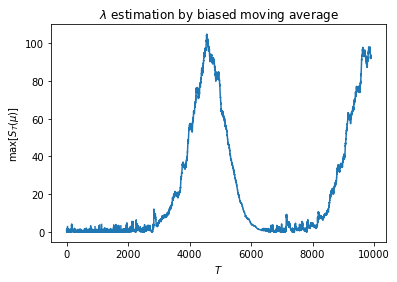}
         \caption{}
         \label{fig:biased_estimation}
     \end{subfigure}
     \hfill
     \begin{subfigure}[b]{0.49\textwidth}
         \centering
         \includegraphics[width=\textwidth]{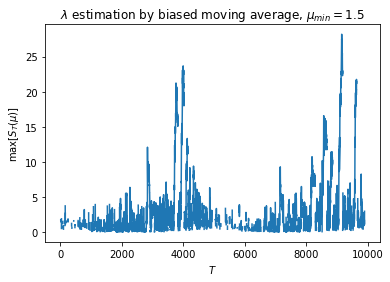}
         \caption{}
         \label{fig:biased_estimation_mumin}
     \end{subfigure}
     
     \caption{Plots of a run of FOCuS over the log-sine wave where various background estimation methods are employed.}
     \label{fig:background_estimations}
\end{figure}

To show this effect, Figure \ref{fig:signal_estimation} shows an oscillating signal drawn using $\lambda(T)$ the exponential of a sine wave, and a non-centred simple moving average allowing estimates of $\lambda(T)$ using only $t \leq T$. A centered moving average that would give an unbiased estimate of $\lambda(T)$ can also be calculated by horizontally shifting the non-centered moving average.

Figure \ref{fig:background_estimations} shows the statistical thresholds recorded by Poisson-FOCuS running over the signal with these different types of background estimation. While Figures \ref{fig:no_estimation} and \ref{fig:unbiased_estimation} show that the unbiased estimation method is comparable to using the true value of $\lambda$, the biased method in Figure \ref{fig:biased_estimation} has large peaks in the recorded statistical threshold as it passes over the signal. These are caused by the upward change in background rate, with lag in the estimation of this, being interpreted as a very small anomaly over a very long time period.

While the effect of a biased estimation method can be somewhat countered by setting a value for $\mu_{\text{min}} > 1$ as in \ref{fig:biased_estimation_mumin}, careful consideration should be given to de-biasing the background estimation method in order to avoid false detections.
\end{document}